\documentclass[a4paper,UKenglish]{lipics_arxiv}

\usepackage{microtype}
\usepackage{xspace}
\usepackage{wrapfig}
\usepackage{tabu}
\usepackage{booktabs}

\graphicspath{{./figures/}}

\newcommand{\etal}{et~al.\xspace}

\title{Topological Stability of Kinetic $k$-Centers
\footnote{W. Meulemans and J. Wulms are (partially) supported by the Netherlands eScience Center (NLeSC) under grant number 027.015.G02. K. Verbeek is supported by the Netherlands Organisation for Scientific Research (NWO) under project no.~639.021.541. Research on the topic of this paper was initiated at the 3rd Workshop on Applied Geometric Algorithms (AGA 2017) in Vierhouten, The Netherlands, supported by the Netherlands Organisation for Scientific Research (NWO) under project no. 639.023.208.}
}

\author[1]{Ivor Hoog v.d.}
\author[1]{Marc van Kreveld}
\affil[1]{Dept. of Information and Computing Sciences, {Utrecht University, The Netherlands}\\
\texttt{[i.d.vanderhoog|m.j.vankreveld]@uu.nl}}
\author[2]{Wouter Meulemans}
\author[2]{Kevin Verbeek}
\author[2]{Jules Wulms}
\affil[2]{Dept. of Mathematics and Computer Science, {TU Eindhoven, The Netherlands}\\
\texttt{[w.meulemans|k.a.b.verbeek|j.j.h.m.wulms]@tue.nl}}

\authorrunning{I. Hoog v.d., M. van Kreveld, W. Meulemans, K. Verbeek and J. Wulms} 

\Copyright{Ivor Hoog v.d., Marc van Kreveld, Wouter Meulemans, Kevin Verbeek and Jules Wulms}


\usepackage{amsmath}
\DeclareMathOperator{\OPT}{OPT}

\DeclareMathOperator{\TS}{\rho_{TS}}



\theoremstyle{plain}
\newtheorem{observation}[theorem]{Observation}

\begin{document}

\maketitle

\begin{abstract}
We study the $k$-center problem in a kinetic setting: given a set of continuously moving points $P$ in the plane, determine a set of $k$ (moving) disks that cover $P$ at every time step, such that the disks are as small as possible at any point in time. Whereas the optimal solution over time may exhibit discontinuous changes, many practical applications require the solution to be \emph{stable}: the disks must move smoothly over time. Existing results on this problem require the disks to move with a bounded speed, but this model allows positive results only for $k<3$. Hence, the results are limited and offer little theoretical insight. Instead, we study the \emph{topological stability} of $k$-centers. Topological stability was recently introduced and simply requires the solution to change continuously, but may do so arbitrarily fast. We prove upper and lower bounds on the ratio between the radii of an optimal but unstable solution and the radii of a topologically stable solution---the topological stability ratio---considering various metrics and various optimization criteria. For $k = 2$ we provide tight bounds, and for small $k > 2$ we can obtain nontrivial lower and upper bounds. Finally, we provide an algorithm to compute the topological stability ratio in polynomial time for constant $k$.

%
 \end{abstract}

\section{Introduction}
The \emph{$k$-center problem} or \emph{facility location problem} asks for a set of $k$ disks that cover a given set of $n$ points, such that the radii of the disks are as small as possible by some measure. The problem can be interpreted as placing a set of $k$ facilities (e.g. stores) such that the distance from every point (e.g. client) to the closest facility is minimized. Since the introduction of the $k$-center problem by Sylvester~\cite{sylvester1857question} in 1857, the problem has been widely studied and has found many applications in practice. Although the $k$-center problem is NP-hard if $k$ is part of the input~\cite{meggido1984complexity}, efficient algorithms have been developed for small $k$. Using rectilinear distance, the problem can be solved in $O(n)$ time~\cite{drezner1987rectangular,hoffmann1999simple,sharir1996rectilinear} for $k = 2,3$ and in $O(n \log n)$ time~\cite{nussbaum1997rectilinear,segal1997piercing} for $k = 4,5$. The problem becomes harder for Euclidean distance, and the currently best known algorithm for Euclidean $2$-centers runs in $O(n \log^2 n (\log \log n)^2)$ time~\cite{chan1999planar}.

In recent decades there has been an increased interest, especially in the computational geometry community, to study problems for which the input points are moving, including the $k$-center problem. These problems are typically studied in the framework of \emph{kinetic data structures}~\cite{basch1999data}, where the goal is to efficiently maintain the (optimal) solution to the problem as the points are moving.
The kinetic version of the $k$-center problem also finds a lot of practical applications in, for example, mobile networks and robotics.
A number of kinetic data structures have been developed for maintaining (approximate) $k$-centers~\cite{degener2010kinetic,friedler2010approximation,gao2003discrete,gao2006deformable}, but in a kinetic setting another important aspect starts playing a role: \emph{stability}.
In many practical applications, e.g., if the disks are represented physically, or if the disks are used for visualization, the disks should move smoothly as the points are moving smoothly. As the optimal $k$-center may exhibit discontinuous changes as points move (see Figure~\ref{fig:discontinuity}), we need to resort to approximations to guarantee stability.

\begin{figure}
\begin{minipage}{0.65\linewidth}
\includegraphics{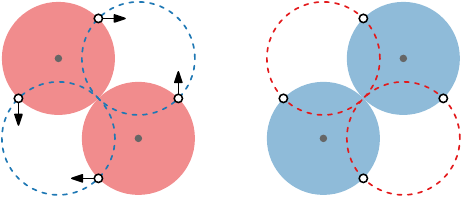}
\end{minipage}	
\begin{minipage}{0.35\linewidth}
\caption{A small change in point positions causes a large change in the smallest 2 covering disks: moving the points along the arrows changes the optimal solution from the red to the blue solution. The optimal centers are marked in grey, while sub-optimal solutions are shown as dashed circles.}
\label{fig:discontinuity}
\end{minipage}
\end{figure}

The natural and most intuitive way to enforce stability is as follows. We assume that each point is moving with at most unit speed, and then we bound the speed of the disks. Durocher and Kirkpatrick~\cite{durocher2008bounded} consider this type of stability for Euclidean $2$-centers. They show that an approximation ratio on the maximum disk radius of $8/\pi \approx 2.55$ can be maintained when the disks can move with speed (at most) $8/\pi + 1 \approx 3.55$.
Similarly, in the black-box KDS model, where the complete input is not known beforehand, de Berg \etal~\cite{deberg2013kinetic} show an approximation ratio of $2.29$ for Euclidean $2$-centers with maximum speed $4\sqrt{2}$.
Furthermore, a similar approach was already used by Bespamyatnikh \etal~\cite{DBLP:conf/dialm/BespamyatnikhBKS00} to find approximations for $1$-centers on moving points. 

However, this natural approach to stability is typically hard to work with and difficult to analyze. This is caused by the fact that several different aspects are influencing the optimality of solutions that move with bounded speed:
\begin{enumerate}
  \item How is the quality of the solution influenced by enforcing continuous motion?
  \item How ``far'' apart are combinatorially different optimal (or approximate) solutions, that is, how long does it take to change from one solution to the other?
  \item How often can optimal (or approximate) solutions change their combinatorial structure?
\end{enumerate}
Ideally we would use a direct approach and design an algorithm that (roughly) keeps track of the optimal solution and tries to stay as close as possible while adhering to the speed constraints. However, especially the latter two aspects make this direct approach hard to analyze. It is therefore no surprise that most (if not all) approaches to stable solutions are indirect: defining a different structure that is stable in nature and that provides an approximation to what we really want to compute. Although interesting in their own right, such indirect approaches have several drawbacks: (1) techniques do not easily extend to other problems, (2) it is hard to perform better (or near-optimal) for instances where the optimal solution is already fairly stable, and (3) these approaches do not offer much theoretical insight into how optimal solutions (or, by extension, approximate solutions) behave as the points are moving. To gain a better theoretical insight into the concept of stability, we need to look at the aspects listed above, ideally in isolation.

\begin{figure}[t]
\centering
\includegraphics{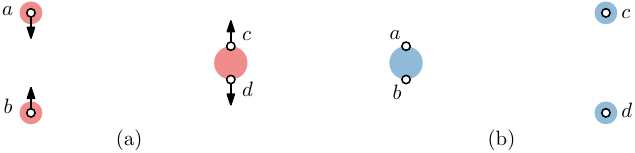}
\caption{An instance for $k=3$ with unbounded Lipschitz stability.}
\label{fig:lipschitz}
\end{figure}

Recently, Meulemans \etal~\cite{meulemans2017framework} introduced a new framework for algorithm stability. This framework includes the natural approach to stability described above (called \emph{Lipschitz stability} in~\cite{meulemans2017framework}), but it also includes the definition of \emph{topological stability}.
An algorithm is topologically stable if its output behaves continuously as the input is changing. The topological stability ratio of a problem is then defined as the optimal approximation ratio of any algorithm that is topologically stable. The analysis of upper and lower bound on this ratio, as in~\cite{meulemans2017framework}, is independent of the problem at hand, and can therefore be applied to not only the kinetic $k$-center problem, but also other variants such as kinetic robust $k$-center, $k$-means, and $k$-median, but also many other unrelated problems. A more formal definition of this type of stability and of continuity are given in the paragraph explaining topological stability.

Due to the fact that it allows arbitrary speed, topological stability is mostly interesting from a theoretical point of view: it provides insight into the interplay between problem instances, solutions, and the optimization function; an insight that is invaluable for the development of stable algorithms. 
Nonetheless, topological stability still has practical uses: an example of a very fast and stable change in visualization can be found when opening a minimized application in most operating systems. The transition starts with the application having a very small size, even as small as a point. The application quickly grows to its intended size in a very smooth and fluid way, which helps the user grasp what is happening.

Next to the results for Lipschitz stability on the 2-center problem mentioned above, Durocher and Kirkpatrick also showed that for $k$-centers with $k > 2$, no approximation factor can be guaranteed with disks of any bounded speed~\cite{durocher2006geometric}.
Figure~\ref{fig:lipschitz} shows an instance that requires infinite speed of the disk centers for $k = 3$. In Figure~\ref{fig:lipschitz} (left) we see point $a$ and $b$ moving closer, and $c$ and $d$ moving apart. When the distance between $a$ and $b$ becomes less than the distance between $c$ and $d$, one of the disks covering $a$ or $b$ has to move to $c$ or $d$, to keep the radius of the disks minimal. The resulting solution can be seen in Figure~\ref{fig:lipschitz} (right). Since the distance between $a,b$ and $c,d$ can be arbitrarily large, either a disk has to move with arbitrarily high speed, or, given bounded speed, the radius of the disk has to be arbitrarily large to ensure that all points are covered at all times. This configuration can occur in any instance of the $k$-center problem for $k \geq 3$, hence the $K$-Lipschitz stability ratio of the $k$-center problem may be unbounded for any bounded speed $K$.

\subparagraph{$k$-center variants}
An instance of the $k$-center problem arises from three choices to obtain variants of the problem: the number $k$ of covering shapes, the geometry of the covering shapes and the criterion that measures solution quality.
In this paper, we consider two types of covering shapes:
(a) in the \emph{Euclidean} model, the covering shapes are disks; (b) in the \emph{rectilinear} model, the covering shapes are axis-aligned squares.
The radius of a covering shape is the distance from its center to its boundary, under $L_2$ for the Euclidean model and $L_{\infty}$ for the rectilinear model.
Furthermore, we distinguish two criteria: (a) in the \emph{minmax} model, the quality of a solution is the maximum radius of its covering shapes, the optimization criterion is to minimize this maximum radius; (b) in the \emph{minsum} model, the quality of a solution is the sum of radii of all $k$ covering shapes, the optimization criterion is to minimize this sum of radii.

The above results in four variants of the problem that can be defined for any $k \geq 2$.
We use the notation $k$-EC and $k$-RC to denote the Euclidean and rectilinear $k$-center problem, appending either -minmax or -minsum to indicate the quality criterion.

\subparagraph{Topological stability}
Let us now interpret topological stability, as proposed in \cite{meulemans2017framework}, for the $k$-centers problem.
Let $\mathcal{I}$ denote the input space of $n$ (stationary) points in $\mathbb{R}^2$ and
$\mathcal{S}^k$ the solution space of all configurations of $k$ disks or squares of varying radii.
Let $\Pi$ denote the $k$-center problem with criterion $f \colon \mathcal{I} \times \mathcal{S}^k \rightarrow \mathbb{R}$ (minmax or minsum).
We call a solution in $\mathcal{S}^k$ valid for an instance in $\mathcal{I}$ if it covers all points of the instance.
An optimal algorithm $\OPT$ maps an instance of $\mathcal{I}$ to a solution in $\mathcal{S}^k$ that is valid and minimizes $f$.

To define instances on moving points and move towards stability, we capture the continuous motion of points in a topology $\mathcal{T}_\mathcal{I}$; an instance of moving points is then a path $\pi \colon [0,1] \rightarrow \mathcal{I}$ through $\mathcal{T}_\mathcal{I}$.
Similarly, we capture the continuity of solutions in a topology $\mathcal{T}^k_\mathcal{S}$, of $k$ disks or squares with continuously moving centers and radii.
A \emph{topologically stable} algorithm $\mathcal{A}$ maps a path $\pi$ in $\mathcal{T}_\mathcal{I}$ to a path in $\mathcal{T}^k_\mathcal{S}$.\footnote{Whereas \cite{meulemans2017framework} assumes the black-box model, we allow omniscient algorithms, knowing the trajectories of the moving points beforehand. That is, the algorithm may use knowledge of future positions to improve on stability. This gives more power to stable algorithms, potentially decreasing the ratio. However, our bounds do not use this and thus are also bounds under the black-box model.}
We use $\mathcal{A}(\pi, t)$ to denote the solution in $\mathcal{S}^k$ defined by $\mathcal{A}$ for the points at time $t$.
The stability ratio of the problem $\Pi$ is now the ratio between the best stable algorithm and the optimal (possibly nonstable) solution:
\begin{equation*}
\TS(\Pi, \mathcal{T}_\mathcal{I}, \mathcal{T}^k_\mathcal{S}) = \inf_{\mathcal{A}} \sup_{\pi \in \mathcal{T}_\mathcal{I}} \sup_{t \in [0,1]} \frac{f(\pi(t), \mathcal{A}(\pi,t))}{f(\pi(t), \OPT(\pi(t)))}
\end{equation*}
where the infimum is taken over all topologically stable algorithms that give valid solutions.
For a minimization problem $\TS$ is at least $1$; lower values indicate better stability.

\begin{table}[b]
    \centering
    \caption{Overview of the bounds on the topological stability ratio for the $k$-center problem.}
    \label{tab:result-table}
\begin{tabu} to 0.8\textwidth {X[0.28,l] X[0.65,l] X[0.23] |[1pt] X[0.8,c] | X[0.9,c] | X[0.9,c]}
  \rowfont[c]{\bfseries}
  \multicolumn{3}{c|[1pt]}{\textbf{$\TS(\Pi, \mathcal{T}_\mathcal{I}, \mathcal{T}^k_\mathcal{S})$}} & $k=2$ & $k=3$ & $k > 3$ \\
  \midrule
  \multirow {4}{=}{\rotatebox[origin=c]{90}{\textbf{Euclidean}}} & \multirow{2}{=}{\textbf{minmax}} & $O$ & \multirow{2}{*}{$\sqrt{2}$} & $\big(1+\sqrt{7}\big)/2$ & 2 \\
  & & $\Omega$ & & $\sqrt{3}$ & $\sin(\frac{\pi (k-1)}{2 k})$ \\
  \cmidrule{2-6}
  & \multirow{2}{=}{\textbf{minsum}} & $O$ & \multirow{2}{*}{2} & \multirow{2}{*}{2} & \multirow{2}{*}{2} \\
  & & $\Omega$ & & & \\
  \midrule
  \multirow {4}{=}{\rotatebox[origin=c]{90}{\textbf{Rectilinear}}} & \multirow{2}{=}{\textbf{minmax}} & $O$ & \multirow{2}{*}{2} & \multirow{2}{*}{2} & \multirow{2}{*}{2} \\
  & & $\Omega$ & & & \\
  \cmidrule{2-6}
  & \multirow{2}{=}{\textbf{minsum}} & $O$ & \multirow{2}{*}{2} & \multirow{2}{*}{2} & \multirow{2}{*}{2} \\
  & & $\Omega$ & & & \\
  \midrule
\end{tabu}
\end{table}

\subparagraph{Contributions}
In this paper we study the topological stability of the $k$-center problem. Although the obtained solutions are arguably not stable, since they can move with arbitrary speed, we believe that analysis of the topological stability ratio offers deeper insights into the kinetic $k$-center problem, and by extension, the quality of truly stable $k$-centers.

In Section~\ref{sec:bounds}, we prove various bounds on the topological stability for this problem, as summarized in Table~\ref{tab:result-table}.
For $k$-EC-minmax, the ratio is $\sqrt{2}$ for $k = 2$; for arbitrary $k$, we prove an upper bound of $2$ and a lower bound that converges to $2$ as $k$ tends to infinity.
For small $k$, we show an upper bound strictly below $2$ as well.
For the other three variants, the stability ratio is exactly $2$ for any $k \geq 2$.
In Section~\ref{sec:algo}, we provide an algorithm to compute a topologically stable solution for an instance of the kinetic $k$-center problem in polynomial time for constant $k$. The approximation ratio for such a solution is generally better than the generic bounds we prove in Section~\ref{sec:bounds}.

\section{Bounds on topological stability}
\label{sec:bounds}
As illustrated above, some point sets have more than one optimal solution.
If we can transform an optimal solution into another, by growing the covering disks or squares by at most (or at least) a factor of $r$, we immediately obtain an upper bound (or respectively a lower bound) of $r$ on the topological stability.
To analyze topological stability of k-center, we therefore start with an input instance for which there is more than one optimal solution, and continuously transform one optimal solution into another. This transformation, which we also call a \emph{morph}, allows the centers to move along a continuous path, while their radii can grow and shrink. At any point during this morph, the intermediate solution should cover all points of the input. In particular, the following property holds: The solution before and after a morph each cover all the points, hence all points live in the intersection of these two solutions. The maximum approximation ratio $r$ that we need for such a morph, gives a bound on the topological stability of $k$-center. We can simply consider the input to be static during the morph, since for topological stability the solution can move arbitrarily fast.
Before analyzing topological stability, we first use the two properties described above, points living in the intersection of two optimal solutions and points being static during morphs, to model relevant problem instances and their optimal solutions in \emph{2-colored intersection graphs}. These graphs allow us to reason about morphs between two optimal solutions.
We then focus on the Euclidean minmax case.
Finally, we briefly consider the minsum and rectilinear cases.

\subparagraph{2-colored intersection graphs}
Consider a point set $P$ and two sets of $k$ convex shapes (disks, squares, ...), such that each set covers all points in $P$: we use $R$ to denote the one set (say, red) and $B$ to denote the other set (blue).
We now define the 2-colored intersection graph $G_{R,B} = (V, E)$: each vertex represents a shape ($V = R \cup B$) and is either red or blue; $E$ contains an edge for each pair of differently colored, intersecting shapes.
A 2-colored intersection graph always contains equally many red nodes as blue nodes, since the nodes of each color represent the $k$ covering shapes in an optimal solution.
Additionally, both solutions represented in a 2-colored intersection graph must cover all points, and hence there may be points only in the area of intersection between a blue and red shape. If there are points that are not covered by one of the two colors, then this color is not a valid solution to the $k$-center problem
In the remainder, we use intersection graph to refer to 2-colored intersection graphs.

\newpage
\begin{lemma}\label{lem:2colortree}
	An intersection forest has at least one node of degree at most~$1$ of each color.
\end{lemma}
\begin{proof}
	Let $F$ be an intersection forest. We prove that $F$ has at least one red node of degree at most 1; the blue case is symmetric.
Since $F$ contains equally many blue and red nodes, there must be a tree $T$ in $F$ having at least as many red nodes as blue nodes.
To arrive at a contradiction, assume that $T$ has only blue leaves.

We decompose $T$ into paths as follows. Pick an arbitrary leaf as a root.
Partition the nodes of $T$ into paths such that each path starts at a nonroot leaf, e.g.\ by running a BFS from each such leaf simultaneously following edges towards the root or using a heavy-path decomposition.
Because $T$ is part of an intersection graph, each path alternates between red and blue nodes.
Hence, the path ending at the root, starting and ending at a blue leaf, has one more blue node than red nodes; the other paths cannot have more red nodes than blue nodes, since at least one endpoint is a leaf and thus blue.
Now, $T$ has more blue than red nodes, which contradicts that $T$ has at least as many red as blue nodes.
Thus, $T$ cannot have only blue leaves.
\end{proof}

\begin{lemma}\label{lem:treestable}
	Consider two sets $R$ and $B$ of $k$ convex translates each covering a point set $P$.
    If intersection graph $G_{R,B}$ is a forest, then $R$ can morph onto $B$ without increasing the shape size, while covering all points in $P$.
\end{lemma}
\begin{proof}
    We prove this lemma by induction on $k$.
    For the base case, $k = 0$, the intersection graph is empty and thus we can trivially morph all red shapes onto the blue shapes.

    For $k > 0$, we reason as follows.
    Since $G_{R,B}$ is a forest, Lemma~\ref{lem:2colortree} tells us that there is a red node $r$ with degree at most 1.
    If $r$ has degree 1, then its one neighbor $b$ must be a blue node; if $r$ has degree 0, then we pick any blue node $b$.
    We morph $r$ onto $b$ by linearly moving the center of $r$ to the center of $b$.
    Since $r$ and $b$ are convex translates, this covers their intersection at all times.
    Now, the new position of the red shape covers any point originally covered by $r$ or $b$.
    Consider $R' = R\setminus\{r\}$ and $B' = B\setminus\{b\}$.
    These sets have size $k-1$ and define an intersection forest $G_{R',B'}$ with $k-1$ shapes.
    The induction hypothesis readily tells us that there is a morph from $R'$ into $B'$ without increasing their size.
    The morph of $r$ onto $b$, followed by the morph of the smaller instance yields us a morph from $R$ to $B$.
\end{proof}

\subparagraph{Euclidean minmax case}
We are now ready to analyze the Euclidean minmax case.
Without loss of generality, we assume here that the disks all have the same radius. If not all disks have the same radius, we can simply grow all disks to have the same radius as the largest disk, since this does not change the optimality nor the validity of a solution.
We first need a few results on (static) intersection graphs, to argue later about topological stability.

\begin{lemma}\label{lem:4cycleUB}
	Let $R$ and $B$ to be optimal solutions to a point set $P$ for $k$-EC-minmax.
	Assume the intersection graph $G_{R,B}$ has a $4$-cycle with a red degree-$2$ vertex.
	To transform $R$ in such a way that $G_{R,B}$ misses one edge of the $4$-cycle, while covering the area initially covered by both sets, it is sufficient to increase the disk radius of a red disk by a factor $\sqrt{2}$.
\end{lemma}
\begin{proof}
To morph from $R$ to $B$, a red disk $r_1$ has to grow to cover the intersection of an adjacent blue disk $b$ with the other (red) neighbor $r_2$ of $b$. Once $r_1$ has grown to overlap the intersection between a blue disk and $r_2$, $r_2$ no longer has to cover this intersection and can be treated as a degree-1 vertex in $G_{R,B}$. A cycle in the intersection graph corresponds to an alternating sequence of red and blue disks, where each consecutive pair has at least one point in its intersection. Applying the above modification will turn two consecutive degree-2 nodes into degree-1 nodes, removing the cycle.
	
	As we have a $4$-cycle of intersections, $a,b,c,d$, we either have to cover both $a$ and $c$ while covering $d$ or $b$, or we have to cover $b$ and $d$ while covering either $a$ or $c$. Let $p_a\in a$ and $p_c\in c$ be the pair of points whose distance is the longest of any pair from $a$ and $c$, and similarly $p_b\in b$ and $p_d\in d$ for $b$ and $d$.
We claim that distance $(p_a,p_c)$ or $(p_b,p_d)$ is shorter than $2\sqrt{2}$.
Assume that distance $(p_a,p_c) > 2 \sqrt{2}$, otherwise we are done.
Since the disks have radius $1$, distances $(p_a,p_b),(p_b,p_c),(p_c,p_d),(p_d,p_a)$ are at most $2$. Our assumption on $(p_a,p_c)$ now implies that the distance from either $p_b$ or $p_d$ to the middle of the line between $p_a$ and $p_c$ is shorter than $\sqrt{2}$. By the triangle inequality, distance $(p_b,p_d)$ is now shorter than $2\sqrt{2}$.
	
	Assume w.l.o.g. that $(p_a,p_c)$ is shorter than $2\sqrt{2}$ and that $p_a$ and $p_b$ are covered by a red disk with only two overlaps. Combining this with the fact that $(p_a,p_b),(p_b,p_c)$ are at most $2$, we can conclude that triple $(p_a,p_b,p_c)$ can be covered by growing the red disk with only two overlaps to radius $\sqrt{2}$.
\end{proof}

\begin{lemma}\label{lem:2colorcycle}
Let $R$ and $B$ be optimal solutions to a point set $P$ for $k$-EC-minmax.
Assume the intersection graph $G_{R,B}$ has only degree-2 vertices.
To transform the disks of $R$ onto $B$, while covering the area initially covered by both sets, it is sufficient to increase the disk radius by a factor $\left(1 + \sqrt{1 + 8 \cos^2(\frac{\pi}{2k})}\right)/2$.
\end{lemma}

\begin{proof}
As the problem is invariant under scaling, we assume w.l.o.g that the maximum radii of the disks is $1$. To arrive are the claimed upper bound, we propose a morph from $R$ to $B$ in which a red disk $r_1$ grows to cover the intersection of an adjacent blue disk $b$ with the other (red) neighbor $r_2$ of $b$.
Specifically, we grow $r_1$ to fully cover its initial disk and the intersection between $b$ and $r_2$ (see the dashed red disk in Figure~\ref{fig:UB-cycle}(a)).
As a result, we now have to consider only $r_1,b,r_2$ without concerning ourselves with the other neighbor of $r_1$ or $r_2$.

\begin{figure}[b]
\centering
\includegraphics[page=1]{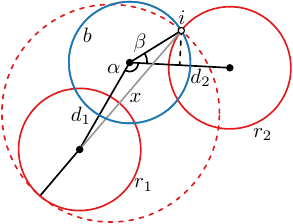}
\caption{Configuration where the smallest angle $\alpha$ occurs in a cycle of overlapping disks. The diameter of red disk $r_1$ after covering its initial area and point $i$ is $1+x$.}
\label{fig:UB-cycle}
\end{figure}

Let $r_1$ be the red disk that has to grow the least, of all red disks in our instance.
Let $0\leq d_1,d_2\leq 2$ be the distance between the centers of $r_1$ and $b$ and between $r_2$ and $b$ respectively.
We know that $d_1 \leq d_2$, as otherwise $r_2$ has to grow less than $r_1$ to cover the other intersection of $b$. However, if $d_2$ is smaller, the intersection between $b$ and $r_2$ is larger so $r_1$ has to grow more to cover the intersection. We can therefore conclude that in the configuration that maximizes the factor by which $r_1$ has to grow, it holds that $d_1 = d_2 = d$.

We use $\alpha$ to denote the angle at the center point of $b$ (see figure).
Disk $r_1$ can be the least-growing disk, only if $\alpha$ is as small as possible.
Larger values of $\alpha$ readily lead to a higher maximum radius for stretching $r_1$.
Since $G_{R,B}$ is a cycle, the $2k$ disk centers thus form the vertices of a simple polygon and we find that the smallest angle in such a polygon is $\alpha \leq \frac{\pi(k-1)}{k}$: the sum of interior angles of a $2k$-gon is $\pi(2k-2)$, hence the smallest angle is never larger than $1/2k$-th of the total angle.
The boundaries of $b$ and $r_2$ intersect in at least one point; we are interested in the point $i$ that is the furthest away from the center point of $r_1$.
Let $\beta$ denote the angle at the center of $b$ between rays towards $i$ and the center of $r_2$.
We know that $\cos(\beta)=d/2$.
The distance $x$ between $i$ and the center of $r_1$ can be found using the Law of Cosines: $x^2=d^2+1^2-2d\cos(\alpha+\beta)$. The diameter of $r_1$ when overlapping its initial area and the intersection between $b$ and $r_2$ is $1+x$.
If the described configuration occurs only with the smallest angle $\alpha$ at a red disk (instead of at a blue disk as shown here), the red disk can grow to overlap both its intersections and fully cover one of the blue disks adjacent to it. This results in a disk with diameter $1+x$ that is sufficient to break the cycle.

Since $d_1 = d_2 = d$, we can symmetrically repeat the above construction to find a point $j$ at the intersection of $r_1$ and $b$, an equivalent of point $i$, under the same angle $\beta$ at the center of $b$. The area enclosed by a line through $i$ and $j$ and disk $b$, containing both intersections with $r_1$ and $r_2$, must enclose a diametrical pair of $b$, otherwise $b$ is not an optimal disk. We therefore find that $\alpha + 2\beta \geq \pi$.
Given our assumption on the radii of the disks and the bound on $\alpha$ we can only ensure that $\alpha + 2\beta > \pi$ by forcing $d$ to become smaller. Consider the triangle with sides $x, y, z$ as in Figure~\ref{fig:UB-cycle}(b). We can show that the length of $x$ decreases as $d$ decreases: looking at the derivatives of $y = d + \cos(\pi - \alpha - \beta)$ and $z = \sin(\pi - \alpha - \beta)$, we see that they are both positive functions for $\alpha \geq \pi/2$, $0 \leq d \leq 2$, and $\beta = \arccos{\frac{d}{2}}$, namely $\frac{\delta y}{\delta d} = 1 - \sin(\pi - \alpha - \arccos{\frac{d}{2}})/\sqrt{1 + (\frac{d}{2})^2}$ and $\frac{\delta z}{\delta d} = \cos(\pi - \alpha - \arccos{\frac{d}{2}})/\sqrt{1 + (\frac{d}{2})^2}$. This means that decreasing $d$ results in a smaller diameter $1+x$ for $r_1$, after growing to cover $i$. To find a worst case ratio between the radii of $r_1$ before and after growing, we therefore get $\alpha + 2\beta = \pi$, resulting in $\beta = \pi\frac{1}{2k}$.
Hence, $\alpha + \beta = \pi - \beta$ and we can derive that $\cos(\alpha + \beta) = \cos(\pi - \beta) = -\cos(\beta)$.
Since $d = 2 \cos(\beta)$ we find that $1 + x \leq 1 + \sqrt{1 + d^2 - 2d \cos(\alpha + \beta)} \leq 1 + \sqrt{1 + 8 \cos^2(\frac{\pi}{2k})}$.
Since $1+x$ is the diameter of $r_1$ after growing, its radius is exactly half this expression.
\end{proof}

\begin{lemma}\label{lem:2colorcycleLB}
	Let $R$ and $B$ be optimal solutions to a point set $P$ for $k$-EC-minmax.
	Assume the intersection graph $G_{R,B}$ has only degree-2 vertices.
	To transform the disks of $R$ onto $B$, while covering the area initially covered by both sets, it may be necessary to increase the disk radius by a factor $2 \sin(\frac{\pi (k-1)}{2 k})$.
\end{lemma}
\begin{proof}
Consider a point set of $2k$ points, positioned such that they are the corners of a regular $2k$-gon with unit radius, i.e., equidistantly spread along the boundary of a unit circle.
There are exactly two optimal solutions for these points (see Fig.~\ref{fig:LB-minmax-k}).
To morph from $R$ to $B$, one of the red disks $r_1$ has to grow to cover the intersection of an adjacent blue disk $b$ with the other (red) neighbor $r_2$ of $b$ (see dashed red disk in Fig.~\ref{fig:LB-minmax-k}). Since the points are all at equal distance from each other on a unit circle, they are the vertices of a regular $2k$-gon. The diameter of the disks in our optimal solution equals the length of a side of this regular $2k$-gon. This means that a red disk has to grow such that its diameter is equal to the distance between a vertex of the $2k$-gon and a second-order neighbor. Hence, the radius of $r_1$ has to grow by a factor of $2 \sin(\frac{\pi (k-1)}{2 k})$.
\qedhere
	\begin{figure}
		\begin{minipage}{0.65\linewidth}
        \centering
		\includegraphics{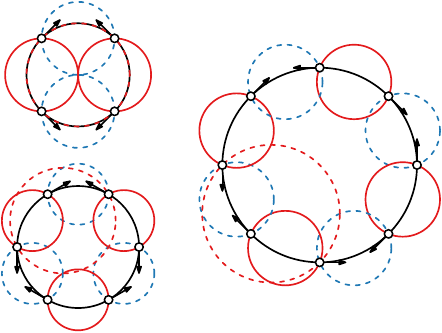}
		\end{minipage}
\begin{minipage}{0.35\linewidth}		
		\caption{Lower bound construction for the stable minmax Euclidean k-center problem, shown for $k=2,3,4$. The optimal solution changes from the solid red to the dashed blue solution. To break the cycle one of the red disks has to grow to the dashed red disk.}
		\label{fig:LB-minmax-k}
		\end{minipage}
	\end{figure}	
\end{proof}

We are now ready to prove bounds on the topological stability of the minmax Euclidean case for moving points. The upcoming sequence of lemmata establishes the following theorem.

\begin{theorem}\label{thm:bounds}
For $k$-EC-minmax, we obtain the following bounds:
\begin{itemize}\setlength{\itemsep}{-\baselineskip}
\item $\TS(2\text{-EC-minmax}, \mathcal{T}_\mathcal{I}, \mathcal{T}^2_\mathcal{S}) = \sqrt{2}$\\
\item $\sqrt{3} \leq \TS(3\text{-EC-minmax}, \mathcal{T}_\mathcal{I}, \mathcal{T}^3_\mathcal{S}) \leq \big(1+\sqrt{7}\big)/2$\\
\item $2 \sin(\frac{\pi (k-1)}{2 k}) \leq \TS(k\text{-EC-minmax}, \mathcal{T}_\mathcal{I}, \mathcal{T}^k_\mathcal{S}) \leq 2$ for $k > 3$.
\end{itemize}
\end{theorem}

\newpage
\begin{lemma}\label{lem:trivupperbound}
	$\TS(k\text{-EC-minmax}, \mathcal{T}_\mathcal{I}, \mathcal{T}^k_\mathcal{S}) \leq 2$ for $k \geq 2$.
\end{lemma}
\begin{proof}
	Consider a point in time $t$ where there are two optimal solutions; let $R$ denote the solution that matches the optimal solution at $t - \varepsilon$ and $B$ the solution at $t + \varepsilon$ for arbitrarily small $\varepsilon > 0$. Let $C$ be the maximum radius of the disks in $R$ and in $B$. Furthermore, let intersection graph $G_{R,B}$ describe the above situation.
	First we make a maximal matching between red and blue vertices that are adjacent in $G_{R,B}$, implying a matching between a number of red and blue disks.
The intersection graph of the remaining red and blue disks has no edges, and we match these red and blue disks in any way. Note that the matched pairs of disks that do not coincide with intersection graph edges do not necessarily intersect.
	
	We find a bound on the topological stability as follows. All the red disks that are matched to blue disks they already intersect grow to overlap their initial disk and the matched blue disk. Now the remaining red disks, which are matched to blue disks they do not necessarily intersect, can safely move to the blue disks they are matched to, and adjust their radii to fully cover the blue disks. Finally, all red disks shrink to match the size of the blue disk they overlap to finish the morph (since each blue disk is now fully covered by the red disk that eventually morphs to be its equal). When all the red disks are overlapping blue disks, the maximum of their radii is at most $2C$, since the radius of each red disk grows by at most the radius of the blue disk it is matched to.
\end{proof}

To prove the lower bound on the topological stability ratio in the next lemma, we must give an instance of the $k$-center problem for which it is not possible to keep a stable solution close to an optimal solution. Since we consider the \emph{kinetic} $k$-center problem, such an instance consists of moving points. We prove a lower bound of $r$ by considering a static instance with two optimal solutions, that requires a stable solution to grow by a factor of $r$ with respect to the optimal solution during a morph. We then show that there is a movement of the points that forces this morph to happen, precisely when the points are in the static configuration we considered. Executing a morph at any other point in time will require the solution to grow more than a factor of $r$.

\begin{lemma}\label{lem:kcenterLB}
	$\TS(k\text{-EC-minmax}, \mathcal{T}_\mathcal{I}, \mathcal{T}^k_\mathcal{S}) \geq 2 \sin(\frac{\pi (k-1)}{2 k})$ for $k \geq 2$.
\end{lemma}
\begin{proof}
The bound readily follows from Lemma~\ref{lem:2colorcycleLB}, if we can show that a set of \emph{moving} points exists that force the solution to change exactly as in Lemma~\ref{lem:2colorcycleLB}.
To this end, consider $2k$ points that move at unit speed along tangents of the unit circle. The tangents touch the unit circle at the corners of an inscribed regular $2k$-gon.
The direction of the points is alternatingly clockwise and counterclockwise with respect to the circle.
We use $t$ to indicate the time at which the moving points are the corners of the inscribed regular $2k$-gon, which are the positions required for Lemma~\ref{lem:2colorcycleLB} (see also Fig.~\ref{fig:LB-minmax-k}).
	
To see that a morph has to happen at time $t$, consider the following. At some time $t'$ before $t$ the pairs of points covered by the red disks in the construction at time $t$ are all coinciding: hence the optimal solution then has maximum radius 0; to not violate our bound, we must have this solution at that time. Morphing $R$ into $B$ between $t'$ and $t$ requires a red disk to grow its radius more than a factor $2 \sin(\frac{\pi (k-1)}{2 k})$, since the red disks are still smaller than at time $t$, while the disks in (combinatorial) solution $B$ are still larger. Analogously, we can argue that we must morph to the blue solution, before a time $t''$ at which the pairs covered by blue disks in the construction coincide. We conclude that the morph has to happen at time $t$ and thus requires the maximum radius to grow by a factor $2 \sin(\frac{\pi (k-1)}{2 k})$ by Lemma~\ref{lem:2colorcycleLB}.
\end{proof}

\begin{lemma}	
    $\TS(2\text{-EC-minmax}, \mathcal{T}_\mathcal{I}, \mathcal{T}^2_\mathcal{S}) = \sqrt{2}$.
\end{lemma}
\begin{proof}
	The lower bound follows directly from Lemma~\ref{lem:kcenterLB} by using $k=2$.
For the upper bound, consider a point in time $t$ where there are two optimal solutions; let $R$ denote the solution that matches the optimal solution at $t - \varepsilon$ and $B$ the solution at $t + \varepsilon$ for arbitrarily small $\varepsilon > 0$.
	If $G_{R,B}$ is a forest, Lemma~\ref{lem:treestable} applies and we do not need to increase the maximum radius during the morph.
	If $G_{R,B}$ contains a cycle, the entire graph must be a $4$-cycle, since in a cycle red and blue alternate and there is at most a single edge between any pair of vertices.
	Lemma~\ref{lem:4cycleUB} gives an upper bound of $\sqrt{2}$ for transforming the intersection graph $G_{R,B}$ to no longer have this $4$-cycle, resulting in a tree. Finally, we can morph $R$ into $B$ without further increasing the maximum radius using Lemma~\ref{lem:treestable}.	
\end{proof}

\begin{lemma}
	$\sqrt{3} \leq \TS(3\text{-EC-minmax}, \mathcal{T}_\mathcal{I}, \mathcal{T}^3_\mathcal{S}) \leq \big(1+\sqrt{7}\big)/2$.
\end{lemma}
\begin{proof}
	Consider a point in time $t$ where there are two optimal solutions; let $R$ denote the solution that matches the optimal solution at $t - \varepsilon$ and $B$ the solution at $t + \varepsilon$ for arbitrarily small $\varepsilon > 0$.
	If intersection graph $G_{R,B}$ is a forest, then Lemma~\ref{lem:treestable} applies and we do not need to increase the maximum radius during the morph.
	If $G_{R,B}$ contains a cycle, then either the entire graph is a $6$-cycle, or there are smaller cycles.
	If the entire graph is a $6$-cycle, the upper bound follows from Lemma~\ref{lem:2colorcycle} and the lower bound from Lemma~\ref{lem:kcenterLB}.
	
	Consider the case where $G_{R,B}$ contains a cycle, but no $6$-cycle. There is at least one $4$-cycle. As $k=3$, every vertex has degree at most $3$. Note that two overlapping disks can be covered by a single disk without increasing the maximum radius beyond $\big(1+\sqrt{7}\big)/2$, if $1+d/2 \leq \big(1+\sqrt{7}\big)/2$, where $d$ is the distance between the centers of the disks. We now distinguish the following cases:
	\begin{itemize}
		\item If there is at most one red degree-$3$ vertex, every $4$-cycle contains at least one degree-$2$ red vertex. Therefore, Lemma~\ref{lem:4cycleUB} can be used to break one of the $4$-cycles by increasing the radius of a red disk by at most $\sqrt{2}$. If $G_{R,B}$ has another $4$-cycle after breaking the first one, we can apply Lemma~\ref{lem:4cycleUB} again. However, if breaking the first $4$-cycle resulted in a $6$-cycle, then the distance between the center points of two adjacent disks in the $4$-cycle was less than $\sqrt{2}<\sqrt{7}-1$. We can therefore fully cover this pair of adjacent disks with the red disk instead of breaking the $4$-cycle. We need a radius of at most $1+\sqrt{2}/2<1+(\sqrt{7}-1)/2=\big(1+\sqrt{7}\big)/2$ for this. Since the red disk now covers all intersections of the blue disk, the resulting intersection graph is a tree.
		\item If there are two or more red degree-$3$ vertices, we look at the distances $d$ between the center points of overlapping disks. If there is a pair of red and blue disks for which $1+d/2 \leq \big(1+\sqrt{7}\big)/2$, we can fully cover the blue disk with the red disk. The remainder of the red disks can now be seen as vertices of degree-$2$ or less in the intersection graph. If there is still a $4$-cycle, then Lemma~\ref{lem:4cycleUB} can be used to break the cycle. If for every pair of red and blue disks $1+d/2 > \big(1+\sqrt{7}\big)/2$ holds,  then the centers of the red disks that\smallskip

\begin{minipage}{0.6\linewidth}
	 overlap the three blue disks can be at most $2-d$ away from each other (see figure). We can cover two red disks $r_1,r_2$ with a single red disk $r_1$ of radius $1+(2-d)/2<1 - \big(1+\sqrt{7}\big)/4<\big(1+\sqrt{7}\big)/2$. We can then freely transform red disk $r_2$ to a blue disk.  Again the remainder of the red disks can now be seen as vertices of degree-$2$ or less in the intersection graph. \phantom{If}
\end{minipage}
\begin{minipage}{0.4\linewidth}
	\centering
	\raisebox{0pt}[\dimexpr\height-1.5\baselineskip\relax]{\includegraphics{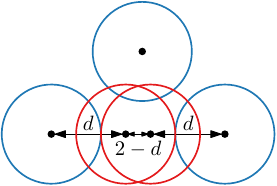}}
\end{minipage}

\vspace{-7.5pt}
If a $4$-cycle remains, then Lemma~\ref{lem:4cycleUB} can be used to break the cycle.
		
	\end{itemize}
	\smallskip In both cases $G_{R,B}$ consists of trees after the changes that were made, thus $R$ can morph into $B$ without further increasing the maximum radius by Lemma~\ref{lem:treestable}.
	In all the above cases we need to grow the maximum radius during the transformation from $R$ to $B$ by at most a factor $\big(1+\sqrt{7}\big)/2$, while in some cases it is necessary to grow to by a factor of $\sqrt{3}$, for example to resolve a $6$-cycle.
    \end{proof}

The above proof shows the strengths and weaknesses of the earlier lemmata. While in many cases we can get close to tight bounds, dealing with high degree vertices in the intersection graph requires additional analysis.
Furthermore, in general we cannot upper bound the approximation factor needed for stable solutions with bounded speed~\cite{durocher2006geometric},  but Theorem~\ref{thm:bounds} act as lower bounds for such bounded speed solutions.

\subparagraph{Rectilinear and minsum cases}
We now turn to the remaining cases.
As it turns out, the stability ratio is $2$ for any value of $k \geq 2$, as captured in the following theorems.

\begin{theorem}\label{thm:RectMinMax}
$\TS(k\text{-RC-minmax}, \mathcal{T}_\mathcal{I}, \mathcal{T}^k_\mathcal{S}) = 2$ for $k \geq 2$.
\end{theorem}
\begin{proof}
The upper bound readily follows from the argument of Lemma~\ref{lem:trivupperbound}.
We prove the lower bound for $k = 2$, understanding that higher values of $k$ cannot lead to a weaker lower bound.

Consider an instance consisting of four points: two points move with unit speed over the lines $y=1$ and $y=-1$ respectively, in opposite directions, while the other two points move with unit speed along the lines $x=-1$ and $x=1$ in opposite directions. Assume that at some time $t$ the points are in the positions $(0,1),(1,0),(0,-1),(-1,0)$. There is exactly one optimum solution for this instance before $t$ and exactly one optimum solution after $t$. However, at time $t$ there are two possible optimum solutions (see Fig.~\ref{fig:LB-minmax-rect-appendix}(a)).
	
To ensure that the squares together cover all points at all times, and that the centers of the squares move in continuous fashion, one of the squares has to grow to cover at least three of the points. After this has happened the second square can move in position of the other optimum solution, followed by shrinking the square that covers three points. To cover at least three points, one of the squares has to grow its radius ($r=1$) to two times the size of the maximum radius of any of the optimum solutions ($r=\frac{1}{2}$) (see Fig.~\ref{fig:LB-minmax-rect-appendix}(b)).
	
	\begin{figure}
		\centering
		\includegraphics{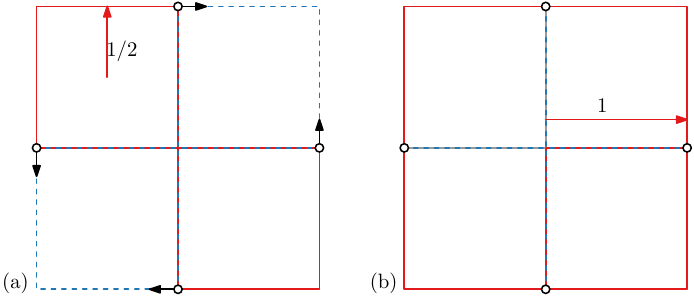}
		\caption{Lower bound construction for the kinetic minmax rectilinear 2-center problem. (a) The optimal solution changes from the solid red to the dashed blue solution, both having radius $1/2$. (b) Continuous transformation requires a maximum radius of $1$, to cover at least three points simultaneously using the large black square.}
		\label{fig:LB-minmax-rect-appendix}
	\end{figure}
	
At any point in time, covering three points with a single square, requires a square of radius at least $1$. Moreover, the radius of the largest square in an optimum solution is smaller than $\frac{1}{2}$ before and after $t$, while this radius is exactly $\frac{1}{2}$ at time $t$. Hence the ratio between the radius of a square covering three points (during any transition) and the optimum radius, is at its smallest at time $t$, resulting a lower bound of $2$ on the topological stability ratio.	
\end{proof}

\newpage
\begin{theorem}\label{thm:MinSum}
$\TS(k\text{-EC-minsum}, \mathcal{T}_\mathcal{I}, \mathcal{T}^k_\mathcal{S}) = \TS(k\text{-RC-minsum}, \mathcal{T}_\mathcal{I}, \mathcal{T}^k_\mathcal{S}) = 2$ for $k \geq 2$.
\end{theorem}
\begin{proof}
The upper bound for the Euclidean case follows from Lemmata~\ref{lem:minsumUB}~and~\ref{lem:minsumLB} below. The proofs only use the triangle inequality and therefore work for general metrics, in particular they work for the rectilinear case.

The lower bound construction for $k = 2$ uses three points: $(0,-1), (0,0), (0,1)$, admitting two optimal solutions, both with a disk of radius $1/2$ and a disk of radius 0 (see Figure~\ref{fig:LB-minsum-eucl-appendix}). Morphing between these requires an intermediate state that double-covers $(0,0)$, or one disk covering all three: the total radius is then $1$. Since the lower bound construction is essentially one dimensional, the disks can be interchanged by squares, and thus works for both the Euclidean and the rectilinear case. 
\end{proof}

\begin{lemma}\label{lem:minsumUB}
	$\TS(k\text{-EC-minsum}, \mathcal{T}_\mathcal{I}, \mathcal{T}^k_\mathcal{S}) \leq 2$ for $k\geq 2$.
\end{lemma}
\begin{proof}
	Consider a point in time $t$ where there are two optimal solutions; let $R$ denote the solution that matches the optimal solution at $t-\varepsilon$ and $B$ the optimal solution at $t+\varepsilon$ for arbitrarily small $\varepsilon>0$. Let $C$ be the sum of radii of $R$, and equivalently the sum of radii of $B$. We morph $R$ onto $B$ while covering all the points with a sum of radii of at most $2C$.
	
	First make a matching between disks in $R$ and disks in $B$. If we look at the intersection graph $G_{R, B}$, we want to create a matching between red and blue vertices that are adjacent. Finding such a matching is easy for cycles since they have as many red disks as blue disks. However, in some cases we might not be able to find a matching between overlapping disks. When $G_{R, B}$ is a forest there can be trees that have more red disks than blue disks and vice versa. In these cases we map the remaining red disks (leaves in $G_{R, B}$) to the remaining blue disks (also leaves in $G_{R, B}$) arbitrarily.
	
	We find a bound on the topological stability as follows. All the red disks that are matched to blue disks they already intersect grow to overlap their initial disk and the blue disk they are matched to. Now the remaining red disks can safely move to the blue disks they are matched to, and adjust their radius to fully cover the blue disks.
Finally, all red disks shrink to match the size of the blue disk they overlap to finish the morph.
When all the red disks are overlapping blue disks, the sum of radii is at most $2C$, since the radius $C_r$ of each red disk $r$ grows by at most the radius $C_b$ of the blue disk $b$ it is matched to.
\end{proof}

\newpage
\begin{lemma}\label{lem:minsumLB}
	$\TS(k\text{-EC-minsum}, \mathcal{T}_\mathcal{I}, \mathcal{T}^k_\mathcal{S}) \geq 2$ for $k\geq 2$.
\end{lemma}
\begin{proof}
	We first construct an instance for $k=2$ that forms the basis of an instance for any $k$. Consider three points on the line $y=0$. One point is at the origin and does not move, while the other two points move at unit speed to the left over $y=0$. Assume that at some time $t$ the points have positions $(-1,0), (0,0), (1,0)$. At any time before and after $t$ there is a single optimal solution. However, at time $t$ there are two optimal solutions; see Fig.~\ref{fig:LB-minsum-eucl-appendix}(a), namely covering one of the outer points and the middle one with a single disk, while covering the remaining point by a radius-$0$ disk.
	
	We want to minimize the sum of the radii of the disks while the points move. If we change from the red to the blue solution at time $t$ in a topologically stable way, we need to cover all three points with a single disk, or grow the smaller disk to also cover the middle point; see Fig.~\ref{fig:LB-minsum-eucl-appendix}(b). In both cases the sum of the radii doubles from $\frac{1}{2}$ to $1$. We cannot preemptively change to the other solution before $t$, since we still need a sum of radii of $1$ during the change, but the optimal solution uses a sum smaller than $\frac{1}{2}$. The case of changing after $t$ is symmetric.
	\begin{figure}
		\centering
		\includegraphics{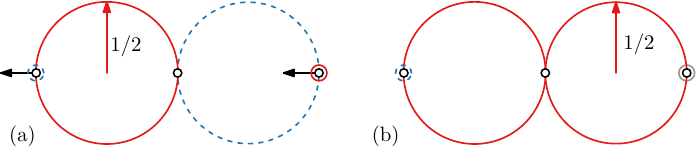}
		\caption{Lower bound construction for the kinetic minsum Euclidean 2-center problem. (a) The optimal solution changes from the solid red solution to the dashed blue solution; the smaller disks have radius zero and thus these solutions have total radius $1/2$. (b) Continuous transformation requires a situation that doubly covers the center point (illustrated) or where one disk covers all three points (not illustrated); either way, the total radius is $1$.}
		\label{fig:LB-minsum-eucl-appendix}
	\end{figure}
	
	The instance for $k=2$ can be extended to hold for any $k$ by having $k-2$ points that are all covered by their own radius-$0$ disk. Placing the remaining points right of $(5,0)$ at least $2$ units away from each other prevents them from being covered by anything but a radius-$0$ disk in an optimal solution.
\end{proof}

\section{Unstable and stable algorithms for $k$-center on moving points}
\label{sec:algo}

Topological stability captures the worst-case penalty that arises from making the
transitions in a solution continuous.
In this section we are interested in the corresponding algorithmic problems that generally result in smaller penalties for a specific instance:
how efficiently can we compute the (unstable) $k$-center for an instance with $n$ moving points, and how efficiently can we compute the stable $k$-center?
When we combine these two algorithms, we can determine for any instance how large the penalty is when solving the given instance in a topologically stable way.
We examine these questions for all four $k$-center variants.

The second algorithm gives us a topologically stable solution to a particular instance of $k$-center. This solution can be used in a practical application requiring stability, for example as a stable visualization of $k$ disks covering the moving points at all time. Since we are dealing with topological stability, the solution can sometimes move at arbitrary speeds. However, in many practical cases, we can alter the solution in a way that bounds the speed of the solution and makes the quality of the $k$-center only slightly worse.

\subsection{Unstable $k$-center algorithms}

Let $P$ be a set of $n$ points moving in the plane, each represented by a constant-degree
algebraic function that maps time to the plane. We denote the point set at time $t$ as $P(t)$ and will develop an algorithm that computes the optimal radius/sum or radii $C$ needed to cover all points with $k$ disks at any point in time in the minmax or minsum model, respectively.

Observe that the minimum covering disks of a point set $P(t)$, denoted $\mathcal{B}^*(t)$, is a set of $k$ disks where each disk is defined by three points in $P(t)$, two points as a diametrical pair,
or a singleton point. In other words, we can define $\mathcal{B}^*$ as the Cartesian product of $k$ triples, pairs, and singletons of distinct points from the set $P(t)$. Not every triple is always relevant: if the circumcircle of the three points is not the boundary of the smallest covering disk,
then the triple is irrelevant at that time. Pairs and singletons always define relevant disks.
This formalization allows us to define what we call \emph{candidate} $k$-centers.

\begin{definition}[Candidate $k$-centers]\label{{def:valid}}
	Any set of $k$ disks $D_1,\ldots,D_k$ where each disk is the minimum covering disk of one, two or three points
	in $P(t)$ is called a \textbf{candidate $k$-center} and is denoted $\mathcal{B}(t)$.
	A candidate $k$-center is \textbf{valid} if the union of its disks cover all points of~$P(t)$.
\end{definition}


This definition allows us to rephrase the goal of the algorithm: For each time $t$ we want to compute the smallest value $C(t)$, such that there exists a valid candidate $k$-center $\mathcal{B}(t)$ where the disks in $\mathcal{B}(t)$ have at most radius $C(t)$ or where their radii sum to $C(t)$ for the \emph{minmax} and \emph{minsum} model respectively.

\subparagraph{Unstable $k$-EC-minmax}
For each singleton, pair, or triple in $P$ we can find the minimum covering disk. Let the radius of this disk be $r$. As the points move along their trajectories, the radius of the minimum covering disk changes over time (unless defined by a singleton point). The function over time giving this radius is continuous for circumcircles because the points that define this radius move continuously. For this reason, the only possibility for discontinuities is when the minimum covering disk of three points changes from being defined by three points to two points or vice versa. However, this change happens only when one of the three points enters/leaves the circumcircle of the two other points. Hence the change between relevant and irrelevant for triples does not introduce discontinuities, as the transition takes place when the circumcircle of the triple coincides with the circumcircle of one pair within the triple.
Each such triple is relevant on $O(1)$ time intervals.
Taking every singleton, pair, and triple of points, we get $O(n^3)$ functions that represent the radii of the minimum covering disks. Any pair of these functions (their images) intersects $O(1)$ times, since the points move along trajectories represented by constant-degree algebraic functions. This implies that the functions form an arrangement of complexity $O(n^6)$. We refer to the images of the functions as curves, for brevity.

Since any pair of curves intersects $O(1)$, in particular the following two functions have this property: consider the function for an arbitrary triple of points $a,b,c$, along with a function for a triple of points that shares exactly two points with $a,b,c$. Whenever the functions intersect, the disks which are represented by the functions have the same radius. Let $p$ be the point in the second triple that differs from the first triple $a,b,c$, then whenever the disks have the same radius, they may coincide and hence all four points considered here may be co-circular. In such a moment, $p$ can enter or exit the circle defined by $a,b$ and $c$. Since this can happen $O(1)$ times, for each point $p$, we can split the function for $a,b,c$ into $O(n)$ pieces where the disk contains the same subset of the input points.

\begin{figure}
    \begin{minipage}{0.6\linewidth}
        \centering
        \includegraphics{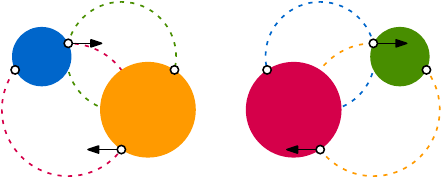}
    \end{minipage}
    \begin{minipage}{0.39\linewidth}
        \centering
        \includegraphics{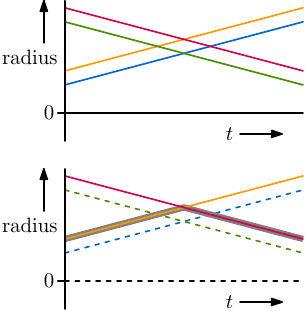}
    \end{minipage}
    \caption{Visualization of the unstable $k$-EC-minmax algorithm (for $k=2$). (Left) A point set is shown for two time steps. The pairs of dashed disks form valid candidate $k$-centers, while the pairs of filled disks form valid candidate $k$-centers with minimal maximum radius. Disks of the same color are defined by the same points. Disks not drawn are either larger than the drawn disks, or require a larger disk to form a valid candidate $k$-center. (Right) The arrangement obtained by plotting the radii of the drawn disks and disks on singleton points. In the bottom arrangement, curves that are not maxcurves are dashed, while the maximum radii of the optimal solution are marked in grey.}
    \label{fig:unstable-algorithm}
\end{figure}

\begin{observation}
	Each of the $O(n^3)$ curves can be split into $O(n)$ pieces where the same subset of points of $P(t)$ are inside the disk corresponding to the curve.
\end{observation}

We are not interested in the arrangement as a whole, but only in the parts where the curves show the maximum radius of a \textbf{valid} $k$-center:  we want to know when the minimum covering disk is the largest of the $k$ disks that cover all the points.
The curve of such a pair or triple may define a part of the solution to the minmax radius problem. For ease of description we will now first continue with the algorithm for the 2-center case, and then show how to extend it to larger values of $k$.

Assume that a pair $a,b$ or triple $a,b,c\in P$ has a minimum covering disk $D_1$ with radius $r_1$ at time $t$.
Let $P_1 \subseteq P$ be all the points covered by $D_1$. To solve the 2-center problem, we need to cover all other points with another disk. Let the minimum covering disk of $P_2=P \setminus P_1$ be $D_2$ with radius $r_2$.
We say that the curve for pair $a,b$ or triple $a,b,c$ is a \emph{maxcurve} at time $t$ if $r_1 \geq r_2$
at time $t$.

A curve can only become a maxcurve at intersections of the curves, since the radii of two covering disks will be equal at an intersection. We refer to a piece of curve between intersections as an \emph{arc}. If we take the arrangement of all maxcurves, we still get an arrangement of complexity $O(n^6)$. It takes $O(n^7)$ time to compute this arrangement, since we need to check if an arc is a maxcurve only at a single time for every arc in the arrangement. As we take all maxcurve arcs, we know that we keep only the parts of the initial arrangement where the maximum radius of a solution is represented. The lower envelope of this arrangement will therefore show the maximum radius of an optimum solution at any point in time for this instance of the Euclidean 2-center problem. See Figure~\ref{fig:unstable-algorithm} for a visualization of an input instance and the corresponding arrangements.

Finding the lower envelope of this arrangement takes $O(\lambda_{s+1}(n^6)\log(n))$ time when every pair of curves intersects at most $s$ times with each other~\cite{hershberger89}, where $\lambda_s(n)$ is the maximum length of a Davenport–Schinzel sequence of order $s$ with $n$ distinct values~\cite{Sharir1995davenport}.
The time needed to compute the lower envelope is dominated by the $O(n^7)$ time we spend on
computing the arrangement of maxcurves itself.

To extend this algorithm to the Euclidean $k$-center, we observe that we can start with the
same set of $O(n^3)$ curves for all singletons, pairs, and triples of moving points.
We define a curve to be \emph{maxcurve}
if the not yet covered points can be covered by at most $k-1$ disks of no larger radius.
For each of the $O(n^6)$ arcs of the arrangement this implies solving a static $(k-1)$-center
problem, which takes $O(n^{2k-1})$~\cite{drezner1984p} or $n^{O(\sqrt{k})}$ time~\cite{hwang93}.

\begin{theorem}
	Given a set of $n$ points whose positions in the plane are determined by constant-degree algebraic functions, the $2$-EC-minmax problem can be solved by an algorithm that runs
	in $O(n^7)$ time. The $k$-EC-minmax problem can be solved in $O(n^{2k+5})$
	or $n^{O(\sqrt{k})}$ time.
\end{theorem}

\subparagraph{Unstable $k$-EC-minsum}
We continue with the minsum version of the Euclidean $k$-center problem.
In this variant we can no longer use maxcurves which define the important radii.
Instead, choose $k$ curves and their corresponding $k$-center, and trace it over time
to determine when the $k$-center covers all points. The number of times when a point
enters or leaves any of the $k$ disks of the $k$-center is $O(kn)=O(n)$. Hence, any
choice of $k$ curves gives $O(n)$ time intervals where the $k$-center is valid.
We sum the $k$ curves (radii) on these intervals to get candidate $k$-center values
and new, summed curves.
In total, there are $O(n^{3k+1})$ new curves that are the sum of $k$ original curves,
and their lower envelope represents the desired function $R(t)$.


\begin{theorem}
	Given a set of $n$ points whose positions in the plane are determined by constant-degree algebraic functions, the $k$-EC-minsum problem can be solved by an
	algorithm that runs in $O(\lambda_{c}(n^{3k+1})\log n)$ time for some constant $c$, where
	$\lambda_c(n)$ is the maximum length of a Davenport–Schinzel sequence of order $c$ with $n$ distinct values.
\end{theorem}


\subparagraph{Unstable $k$-RC-minmax}
The rectilinear version of the $k$-center problem in the minmax model is solved by similar methods
as the Euclidean version, albeit simpler and more efficient. We use pairs of points
to define curves that define the radius over time by letting the points be on opposite sides of
a smallest covering square.
This square is not unique, as the points on opposite sides block movement only in a single direction. Thus, there are $O(n)$ different subsets of points that can be covered by such a
square. We first consider the $2$-center case,
where it is sufficient to take the two extreme squares: two squares that cover points with extremal coordinates, for example one square covers the topmost point while the other covers the bottommost point. Extreme squares are sufficient, because in the case where the second square must cover points that lie beyond two opposite sides of the first square, then that second square must be larger. Therefore, in comparison the solution using two extreme squares has a smaller or equal maximum radius.

In total we have $O(n^2)$ curves, defined by pairs of points on opposite sides of a square, which form an arrangement of complexity $O(n^4)$. We again use the concept of maxcurves: we are interested in those arcs of the arrangement
for which the not yet covered points can be covered by a square of no larger size.
While we can test this easily in
linear time for each of the $O(n^4)$ arcs, we can use the arrangement to do this faster. For each curve $C$ corresponding to a square $S$, we process the points moving
in and out of $S$ and maintain the leftmost, rightmost, bottommost and topmost uncovered points, for example using heaps. As the points move, each square will start and stop covering other points $O(n)$ times. The uncovered points can swap places $O(n^2)$ times in an ordering of their $x$- and $y$-coordinates. Thus we maintain the left-, right-, top-, and bottommost points in $O(n^2 \log n)$ for each curve $C$, and determine in constant time whether each of the $O(n^2)$ arcs along the curve is a maxcurve.
Finally, we compute the lower envelop of the maxcurves in time linear in the number of maxcurvce, since the maxcurves are disjoint arcs.

For $k$-centers, the observation that one can use the two extreme squares
for a pair of points no longer holds.
We therefore define squares by triples as in the Euclidean case, the main difference being that the $O(n^3)$
curves we get have considerable overlap: there must still be two defining points on opposite sides of the covering squares defined by triples, hence all $O(n)$ squares that have these exact points on opposite sides must have the same area. The arrangement of the $O(n^3)$ curves is therefore equivalent to the arrangement for the $O(n^2)$ curves defined above, and each arch of the arrangement can be part of $O(n)$ curves. Since the $O(n^2)$ curves form an arrangement of complexity $O(n^4)$, we need to test $O(n^5)$ arcs for being maxcurves. Following the analysis as in the Euclidean case, we obtain:

\begin{theorem}
	Given a set of $n$ points whose positions in the plane are determined by constant-degree algebraic functions, the $2$-RC-minmax problem can be solved by an algorithm that runs
	in $O(n^4\log n)$ time. The $k$-RC-minmax problem can be solved in $O(n^{2k+4})$
	or $n^{O(\sqrt{k})}$ time.
\end{theorem}

\subparagraph{Unstable $k$-RC-minsum}
We can use exactly the same approach as in the Euclidean case, and obtain:

\begin{theorem}
	Given a set of $n$ points whose positions in the plane are determined by constant-degree algebraic functions, the $k$-RC-minsum problem can be solved by an
	algorithm that runs in $O(\lambda_{c}(n^{3k+1})\log n)$ time for some constant $c$, where
	$\lambda_c(n)$ is the maximum length of a Davenport–Schinzel sequence of order $c$ with $n$ distinct values.
\end{theorem}

\subsection{Topologically stable $k$-center algorithms}

In this section we describe an algorithm to compute a topologically stable solution for the Euclidean $k$-center variants.
We use only combinatorial properties of solutions, which also hold for the rectilinear $k$-center variants. Hence the same algorithm, replacing disks by squares, also solves the rectilinear variants.

Intuitively, the unstable algorithm finds the lower envelope of all the \emph{valid} radii by traversing the arrangement of all valid radii over time. At each time $t$ a minimal enclosing disk $D_1$ (defined by a set of at most three points) in the set of optimal disks $\mathcal{B}^*(t)$ needs to be replaced with another disk $D_2$, we ``hop'' from our previous curve to the curve corresponding to the new disk $D_2$. If we require that the algorithm is topologically stable these hops have a cost associated with them.

We first show how to model and compute the cost $C(t)$ of a topological transition between any two $k$-centers at a fixed time $t$. We then extend this approach to work over time. Let $t$ be a fixed moment in time where we want to go from one $k$-center $\mathcal{B}_1 $ to another candidate $k$-center $\mathcal{B}_2$. The transition can happen at infinite speed but must be continuous. We denote the infinitesimal time frame around $t$ in which we do the transition as
$[0,T]$.
We extend the concept of a $k$-center with a corresponding partition of the input points at time $t$ $P(t)$ over the disks in the $k$-center:

\begin{definition}[Disk set]
For each disk $D_i$ of a candidate $k$-center $\mathcal{B}$ for $P(t)$ we define its \textbf{disk set} $P_i\subseteq P(t)\cap D_i$ as the subset of points assigned to $D_i$.
A candidate $k$-center $\mathcal{B}$ with disk sets $P_1,\ldots,P_k$ is \textbf{valid} if the disk sets partition $P(t)$.
We say $\mathcal{B}$ is \textbf{valid} if there exist disk sets $P_1,\ldots,P_k$ such that
$\mathcal{B}$ with disk sets $P_1,\ldots,P_k$ is valid.
\end{definition}

During a continuous transformation, the disk sets of a valid candidate $k$-center will change in the time interval $[0,T]$ while the points $P(t)$ do
not move. In essence the time $t$ is equivalent to the whole interval $[0,T]$. For ease of understanding,
we use $t'$ to denote any time in the interval $[0,T]$.
Observe that our definition of topological stability leads to an intuitive way of recognizing a stable transition:

\begin{lemma}
	\label{lemma:equivalence}
	A transition from one candidate $k$-center $\mathcal{B}_1(t)$ to another candidate $k$-center $\mathcal{B}_2(t)$ in the time interval $[0,T]$ is \textbf{topologically stable} if and only if the change of the disks' centers and radii is continuous over $[0,T]$ and at each time $t' \in [0,T]$, $\mathcal{B}(t')$ is \textbf{valid}.
\end{lemma}

\begin{proof}
	Note that by definition the disks must be transformed continuously and that all the points in $P(t)$ are covered in $[0,T]$ precisely when a valid candidate $k$-center exists.
\end{proof}

Now that we can recognize a topologically stable transition, we can reason about what such a transition looks like:

\begin{lemma}
	\label{lemma:interleave}
	Any topologically stable transition from one $k$-center $\mathcal{B}_1 (t)$ to another $k$-center $\mathcal{B}_2 (t)$ in the timespan $[0, T]$ that minimizes $C(t)$ (the largest occurring \emph{minsum/minmax} over $[0,T]$) can be obtained by a sequence of events where in each event, occurring at a time $t'\in [0,T]$, a disk $D_i \in \mathcal{B}(t')$ adds a point to $P_i$ and a disk $D_j \in \mathcal{B}(t')$ removes a point from $P_j$. We call this \textbf{transferring}.
\end{lemma}
\begin{proof}
	The proof is by construction. Assume that we have a transition from  $\mathcal{B}_1(t)$ to  $\mathcal{B}_2(t)$ and that the transition that minimizes either the sum or the maximum of all radii contains simultaneous continuous movement. Let this transition take place in $[0,T]$.
	
	To determine $C(t)$ we need to look only at times $t' \in [0,T]$ where a disk $D_i \in \mathcal{B}$ adds a new point $p$ to its disk set $P_i$ and another disk $D_j$ removes it from $P_j$. Only at $t'$ must both disks contain $p$; before $t'$ disk $D_j$ may be smaller and after $t'$ disk $D_i$ may be smaller.

	We claim that any optimal simultaneous continuous movement of cost $C(t)$ can be discretized into a sequence of events with cost no larger than $C(t)$. We do so recursively: for the continuous movement there exists a $t_0 \in [0,T]$ as the first time a disk $D_i \in \mathcal{B}$ adds a point $p$ to $P_i$. Then at $t_0$, $D_i$ has to contain both $P_i$ and $p$ and must have a certain size $d$. All the other disks $D_j \in \mathcal{B}$ with $j \neq i$ only have to contain the points in $P_j$ so they have optimal size if they have not moved from time $0$. In other words, it is optimal to first let $D_i$ obtain $p$ in an event and to then continue the transition from $[t_0,T]$. Applying this argument recursively allows us discretize the simultaneous movement into sequential events.
\end{proof}

\begin{corollary}
	\label{corollary:sequence}
	Any topologically stable transition from one $k$-center $\mathcal{B}_1 (t)$ to another $k$-center $\mathcal{B}_2(t)$  in the timespan $[0, T]$ that minimizes $C(t)$ (the largest occurring \emph{minsum/minmax} over $[0,T]$) can be obtained by a sequence of events where in each event the following happens:
	\begin{itemize}
		\item A disk $D_i \in \mathcal{B}_1(t)$ that was defined by one, two or three points in $P(t)$ is now defined by a new set of points in $P(t)$ where the two sets differ in only one element.
	\end{itemize}
	With every event, $P_i$ must be updated with a corresponding insert and/or delete. We call these events a \textbf{swap} because we intuitively swap one of the defining elements.
\end{corollary}


\subparagraph{The cost of a single stable transition}
Corollary \ref{corollary:sequence} allows us to model a stable transition as a sequence of swaps but how do we find the optimal sequence of swaps? A single minimal covering disk at time $t$ is defined by at most three unique elements from $P(t)$ so there are at most ${O}(n^3)$ subsets of $P(t)$ that could define one disk of a $k$-center. Let these ${O}(n^3)$ sets be the vertices in a graph $G$. We create an edge between two vertices $v_i$ and $v_j$ if we can transition from one disk to the other with a single swap and that transition is topologically stable.
Each vertex is incident to only a constant number of edges (apart from degenerate cases) because during a swap the disk $D_i$ corresponding to $v_i$ can only add one element to $P_i$.
Moreover, the radius of the disk is maximal on vertices in $G$ and not on edges. The graph $G$ has ${O}(n^3)$ complexity and takes ${O}(n^4)$ time to construct, as we need to check for each vertex which edges should be added. This can be done by checking all $O(n)$ vertices for which the disk set differs by one element.

\begin{figure}[t]
	\centering
	\includegraphics{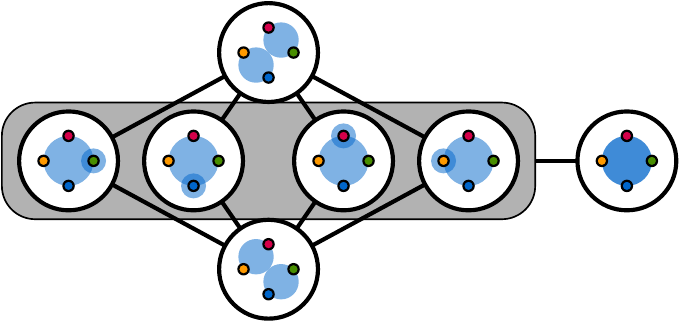}
	\caption{The $G^2$ graph on four points. Vertices of $G^2$ in the grey area all have an edge to the vertex on the right. This rightmost vertex has a disk defined by the red and blue points and another disk defined by the yellow and green points.}
	\label{fig:transition-graph}
\end{figure}

This graph provides a framework to trace the radius of the transition from a single disk to another disk. However, we want to transition from one $k$-center to another. We use the previous graph to construct a new graph $G^k$ where each vertex $w_i$ represents a set of $k$ disks: a candidate $k$-center $\mathcal{B}_i$. We again create an edge between vertices $w_i$ and $w_j$ if we can go from the candidate $k$-center $\mathcal{B}_i$ to $\mathcal{B}_j$ in a single swap. With a similar argument as above, each vertex is only connected to ${O}(k)$ edges. The graph thus has ${O}(n^{3k})$ complexity and can be constructed in ${O}(n^{3k+1})$ time. This time, each of the $O(n^{3k})$ vertices has to check $O(kn)=O(n)$ other vertices to determine whether there should be an edge: for each of the $k$ disks there are $O(n)$ vertices for which the disk set differs by one element. Each vertex $w_i$ gets assigned the cost (minmax/minsum) of the $k$-center $\mathcal{B}_i$ where the cost is $\infty$ if $\mathcal{B}_i$ is invalid. Figure~\ref{fig:transition-graph} shows an example of a $G^2$ graph.

Any connected path in this graph from $w_i$ to $w_j$ without vertices with cost $\infty$ represents a stable transition from $w_i$ to $w_j$ by Corollary \ref{corollary:sequence}, where the cost of the path (transition) is the maximum value of the vertices on the path.
We can now find the optimal sequence of swaps to transition from any vertex $w_i$ to $w_j$ by finding the cheapest path in this graph in  ${O}(n^{3k}\log n)$ time, for examply using an adapted version of Dijkstra's algorithm to maintain the highest cost of the cheapest path. The running time is dominated by the ${O}(n^{3k+1})$ time it takes to construct the graph.

\subparagraph{Maintaining the cost of a flip}
For a single point in time we can now determine the cost of a topologically stable transition from a $k$-center $\mathcal{B}_i$ to $\mathcal{B}_j$ in ${O}(n^{3k+1})$ time.  If we want to maintain the cost $C(t)$ for all times $t$, the costs of the vertices in the graph change over time.
If we plot the changes of these costs over time, the graph consists of monotonously increasing or decreasing segments, separated by moments in time where two radii of disks are equal.
A straightforward counting results in ${O}(n^{3k})$ of such events, dominated by splitting the cost function of each vertex in monotone segments. These events also contain all events where the structure of our graph $G^k$ changes \emph{and} all the moments where a vertex in our graph becomes \emph{invalid} and thus gets cost $\infty$, which happens when four points become co-circular. The result of these observations is that we have a ${O}(n^{3k})$ size graph, with ${O}(n^{3k})$ relevant changes, where with each change we spend $O(kn) = {O}(n)$ time per vertex to restore the graph by recomputing the edge set. We can then compute an optimal solution as described in the previous paragraph, at each event update $G^k$ and use the updated graph to find the cost of a continuous transition, if the event coincides with a discrete change in the optimal solution. This leads to an algorithm which can determine the cost of a topologically stable movement in ${O}(n^{6k+1})$ time.

\begin{theorem}
Given a set of $n$ points whose positions in the plane are determined by constant-degree algebraic functions, the cost of a topologically stable solution to the $k$-EC-minmax/-minsum problem can be found by an algorithm that runs in $O(n^{6k+1})$ time.
\end{theorem}

If we run the unstable and stable algorithms on the moving points, we obtain two functions that map time to a cost. The ratio of these two costs over time corresponds to the approximation ratio of the computed topologically stable solution at any point in time. The maximum of this ratio over time is the topological stability ratio of this solution, which is therefore obtained in $O(n^{6k+1})$ time as well. Since the stable algorithm is constructive, we can also find a topologically stable solution in the same time, by returning the set of disks corresponding to the minimum cost path through the time-varying graph $G^k$.


\section{Conclusion}

We considered the topological stability of common variants of the kinetic $k$-center problem, in which solutions must change continuously but may do so arbitrarily fast.
We have established tight bounds for the minsum case (Euclidean and rectilinear) as well as for the rectilinear minmax case for any $k \geq 2$.
For the Euclidean minmax case, we proved nontrivial upper bounds for small values of $k$ and presented a general lower bound tending towards $2$ for large values of $k$.
We also presented algorithms to compute topologically stable solutions together with the cost of stability for a set of moving points, that is, the growth factor that we need for that particular set of moving points at any point in time. A practical application of these algorithms would be to identify points in time where we could slow down the solution to explicitly show stable transitions between optimal solutions.

\subparagraph{Future work}
It remains open whether a general upper bound strictly below 2 is achievable for $k$-EC-minmax.
We conjecture that this bound is indeed smaller than 2 for any constant $k$, since this is clearly the case for configurations where where the intersection graph forms a tree or cycle.
What remains is to acquire more insight in how to resolve an intersection graph with more general structures, in particular how combinatorial properties can help resolve high degree nodes.
Our algorithms to compute the cost of stability for an instance have high (albeit polynomial) run-time complexity. Can the results for KDS (e.g. \cite{deberg2013kinetic}) help us speed up these algorithms? Alternatively, can we approximate the cost of stability more efficiently?

Lipschitz stability requires a bound on the speed at which a solution may change \cite{meulemans2017framework}.
This stability for $k > 2$ is unbounded, if centers have to move continuously \cite{durocher2006geometric}, as Figure~\ref{fig:lipschitz} shows in the introduction.
A potentially interesting variant of the solution space topology, is one where a disk may shrink to radius 0, at which point it disappears and may reappear at another location.
This alleviates the problem in the example; would it allow us to bound the Lipschitz stability?

\clearpage
\bibliography{references}

\end{document}